\def\final{0}
\def\trim{0}

\documentclass[11pt]{article}

\usepackage{dsfont}
\usepackage{algorithm, algorithmic}
\usepackage{amsmath, amssymb, amsthm}
\usepackage{framed}
\usepackage{fullpage}
\usepackage{graphicx}
\usepackage{verbatim}
\usepackage{color}
\definecolor{DarkGreen}{rgb}{0.1,0.5,0.1}
\definecolor{DarkRed}{rgb}{0.5,0.1,0.1}
\definecolor{DarkBlue}{rgb}{0.1,0.1,0.5}
\usepackage{enumitem}
\usepackage[pdftex]{hyperref}
\hypersetup{
    unicode=false,				
    pdftoolbar=true,				
    pdfmenubar=true,			
    pdffitwindow=false, 		
    pdftitle={},    					
    pdfauthor={}
    pdfsubject={},				
    pdfnewwindow=true,		
    pdfkeywords={},			
    colorlinks=true,				
    linkcolor=DarkRed,		
    citecolor=DarkGreen,	
    filecolor=DarkRed,		
    urlcolor=DarkBlue,		
}

\floatstyle{boxed}
\restylefloat{figure}

\ifnum\final=0
\newcommand{\mynote}[1]{\marginpar{\tiny \sf #1}}
\else
\newcommand{\mynote}[1]{}
\fi

\ifnum\trim=0
\newcommand{\trimrm}[1]{#1}
\newcommand{\trimtext}[1]{}
\else
\newcommand{\trimrm}[1]{}
\newcommand{\trimtext}[1]{#1}
\fi

\newcommand\N{\mathbb{N}}

\newcommand{\cA}{\mathcal{A}}
\newcommand{\cB}{\mathcal{B}}
\newcommand{\cC}{\mathcal{C}}
\newcommand{\cD}{\mathcal{D}}

\newcommand{\cH}{\mathcal{H}}

\newcommand{\cK}{\mathcal{K}}

\newcommand{\cQ}{\mathcal{Q}}

\newcommand{\cS}{\mathcal{S}}

\newcommand{\cX}{\mathcal{X}}

\newcommand{\eps}{\varepsilon}
\newcommand{\poly}{\mathrm{poly}}

\newcommand{\getsr}{\gets_{\mbox{\tiny R}}}

\newcommand{\negl}{\mathrm{negl}}

\newtheorem{theorem}{Theorem}[section]
\newtheorem{lemma}[theorem]{Lemma}

\newtheorem{corollary}[theorem]{Corollary}
\newtheorem{proposition}[theorem]{Proposition}

\theoremstyle{definition}
\newtheorem{definition}[theorem]{Definition}

\newcommand{\error}{{\rm error}}

\newcommand{\point}{\mathsf{Point}}

\newcommand{\thresh}{\mathsf{Thresh}}

\def\poly{\mathop{\rm{poly}}\nolimits}

\newcommand{\pos}{\operatorname{POS}}
\newcommand{\encthresh}{\mathsf{EncThresh}}
\newcommand{\stat}{\mathsf{STAT}}
\newcommand{\sig}{\mathsf{ValidSig}}

\newcommand{\hashpoint}{\mathsf{HashPoint}}

\newcommand{\Gen}{\mathsf{Gen}}
\newcommand{\Sign}{\mathsf{Sign}}
\newcommand{\Ver}{\mathsf{Ver}}

\newcommand{\Setup}{\mathsf{Setup}}

\newcommand{\Enc}{\mathsf{Enc}}
\newcommand{\Dec}{\mathsf{Dec}}
\newcommand{\PRF}{\mathsf{PRF}}
\newcommand{\PRG}{\mathsf{PRG}}
\newcommand{\Comp}{\mathsf{Comp}}

\newcommand{\Com}{\mathsf{Com}}
\newcommand{\IO}{\mathsf{IO}}
\newcommand{\GenTT}{\mathsf{Gen_{ex}}}
\newcommand{\TraceTT}{\mathsf{Trace_{ex}}}
\newcommand{\Prove}{\mathsf{Prove}}

\newcommand{\sk}{\mathsf{sk}}

\newcommand{\params}{\mathsf{params}}
\newcommand{\crs}{\mathsf{crs}}
\newcommand{\vk}{\mathsf{vk}}

\newcommand{\accept}{\mathsf{accept}}
\newcommand{\reject}{\mathsf{reject}}

\title{Order-Revealing Encryption and the \\ Hardness of Private Learning}
\author{	Mark Bun\thanks{School of Engineering \& Applied Sciences, Harvard University. \texttt{mbun@seas.harvard.edu}. Supported by an NDSEG fellowship and NSF grant CNS-1237235.} \qquad
		Mark Zhandry\thanks{Stanford University.  \texttt{mzhandry@gmail.com}.  Supported by the DARPA PROCEED program.}
} 
\date{May 2, 2015}

\begin{document}

\maketitle

\begin{abstract}
An order-revealing encryption scheme gives a public procedure by which two ciphertexts can be compared to reveal the ordering of their underlying plaintexts. We show how to use order-revealing encryption to separate computationally efficient PAC learning from efficient $(\eps, \delta)$-differentially private PAC learning. That is, we construct a concept class that is efficiently PAC learnable, but for which every efficient learner fails to be differentially private. This answers a question of Kasiviswanathan et al. (FOCS '08, SIAM J. Comput. '11).

To prove our result, we give a generic transformation from an order-revealing encryption scheme into one with strongly correct comparison, which enables the consistent comparison of ciphertexts that are not obtained as the valid encryption of any message. We believe this construction may be of independent interest.
\end{abstract}
\thispagestyle{empty}

\vfill


\newpage

\setcounter{page}{1}

\section{Introduction}

\label{sec:intro}

Many agencies hold sensitive information about individuals, where statistical analysis of this data could yield great societal benefit. The line of work on differential privacy \cite{DworkMcNiSm06} aims to enable such analysis while giving a strong formal guarantee on the privacy afforded to individuals. Noting that the framework of computational learning theory captures many of these statistical tasks, Kasiviswanathan et al. \cite{KLNRS11} initiated the study of \emph{differentially private learning}. Roughly speaking, a differentially private learner is required to output a classification of labeled examples that is accurate, but does not change significantly based on the presence or absence of any individual example.

The early positive results in private learning established that, ignoring computational complexity, any concept class is privately learnable with a number of samples logarithmic in the size of the concept class \cite{KLNRS11}. Since then, a number of works have improved our understanding of the sample complexity -- the minimum number of examples -- required by such learners to simultaneously achieve accuracy and privacy. Some of these works showed that privacy incurs an inherent additional cost in sample complexity; that is, some concept classes require more samples to learn privately than they require to learn without privacy \cite{BeimelKaNi10, CH11, BeimelNiSt13b, FX14, ChaudhuriHsSo14, BunNiStVa15}. In this work, we address the complementary question of whether there is also a \emph{computational} price of differential privacy for learning tasks, for which much less is known. The initial work of Kasiviswanathan et al. \cite{KLNRS11} identified the important question of whether any efficiently PAC learnable concept class is also efficiently privately learnable, but only limited progress has been made on this question since then \cite{BeimelKaNi10, Nissim14}. 

Our main result gives a strong negative answer to this question. We exhibit a concept class that is efficiently PAC learnable, but under plausible cryptographic assumptions cannot be learned efficiently and privately. To prove this result, we establish a connection between private learning and \emph{order-revealing encryption}. We construct a new order-revealing encryption scheme with strong correctness properties that may be of independent learning-theoretic and cryptographic interest.

\subsection{Differential Privacy and Private Learning}

We first recall Valiant's (distribution-free) PAC model for learning \cite{Valiant84}. Let $\cC$ be a \emph{concept class} consisting of concepts $c : X \to \{0, 1\}$ for a data universe $X$. A learner $L$ is given $n$ samples of the form $(x_i, c(x_i))$ where the $x_i$'s are drawn i.i.d. from an unknown distribution, and are labeled according to an unknown concept $c$. The goal of the learner is to output a \emph{hypothesis} $h : X \to \{0, 1\}$ from a hypothesis class $\cH$ that approximates $c$ well on the unknown distribution. That is, the probability that $h$ disagrees with $c$ on a fresh example from the unknown distribution should be small -- say, less than $0.05$. The hypothesis class $\cH$ may be different from $\cC$, but in the case where $\cH \subseteq \cC$ we call $L$ a \emph{proper} learner. Moreover, we say a learner is \emph{efficient} if it runs in time polynomial in the description size of $c$ and the size of its examples.

Kasiviswanathan et al. \cite{KLNRS11} defined a private learner to be a PAC learner that is also differentially private. 
Two samples $S = \{(x_1, b_1), \dots, (x_n, b_n)\}$ and $S' = \{(x'_1, b'_1), \dots, (x'_n, b'_n)\}$ are said to be \emph{neighboring} if they differ on exactly one example, which we think of as corresponding to one individual's information. A randomized learner $L : (X \times \{0, 1\})^n \to \cH$ is  $(\eps, \delta)$-\emph{differentially private} if for all neighboring datasets $S$ and $S'$ and all sets $T \subseteq \cH$, 
\[\Pr[L(S) \in T] \le e^{\eps}\Pr[L(S') \in T] + \delta.\]
The original definition of differential privacy \cite{DworkMcNiSm06} took $\delta = 0$, a case which is called \emph{pure} differential privacy. The definition with positive $\delta$, called \emph{approximate} differential privacy, first appeared in \cite{DKMMN06} and has since been shown to enable substantial accuracy gains. Throughout this introduction, we will think of $\eps$ as a small constant, e.g. $\eps = 0.1$, and $\delta = o(1/n)$.

Kasiviswanathan et al. \cite{KLNRS11} gave a generic ``Private Occam's Razor'' algorithm, showing that any concept class $\cC$ can be privately (properly) learned using $O(\log |\cC|)$ samples. Unfortunately, this algorithm runs in time $\Omega(|\cC|)$, which is exponential in the description size of each concept. With an eye toward designing efficient private learners, Blum et al. \cite{BlumDwMcNi05} made the powerful observation that any efficient learning algorithm in the \emph{statistical queries} (SQ) framework of Kearns \cite{Kearns98} can be efficiently simulated with differential privacy. Moreover, Kasiviswanathan et al. \cite{KLNRS11} showed that the efficient learner for the concept class of parity functions based on Gaussian elimination can also be implemented efficiently with differential privacy. These two techniques -- SQ learning and Gaussian elimination -- are essentially the only methods known for computationally efficient PAC learning. The fact that these can both be implemented privately led Kasiviswanathan et al. \cite{KLNRS11} to ask whether \emph{all} efficiently learnable concept classes could also be efficiently learned with differential privacy.

Beimel et al. \cite{BeimelKaNi10} made partial progress toward this question in the special case of pure differential privacy with proper learning, showing that the sample complexity of efficient learners can be much higher than that of inefficient ones. \trimrm{Specifically, they showed that assuming the existence of pseudorandom generators with exponential stretch, there exists for any $\ell(d) = \omega(\log d)$ a concept class over $\{0, 1\}^d$ for which every efficient proper private learner requires $\Omega(d)$ samples, but an inefficient proper private learner only requires $O(\ell(d))$ examples.} Nissim \cite{Nissim14} strengthened this result substantially for ``representation learning,'' where a proper learner is further restricted to output a canonical representation of its hypothesis. He showed that, assuming the existence of one-way functions, there exists a concept class that is efficiently representation learnable, but not efficiently privately representation learnable (even with approximate differential privacy). With Nissim's kind permission, we give the details of this construction in \trimrm{Section \ref{sec:proper}}\trimtext{the full version of this work}.

Despite these negative results for proper learning, one might still have hoped that any efficiently learnable concept class could be efficiently \emph{improperly} learned with privacy. Indeed, a number of works have shown that, especially with differential privacy, improper learning can be much more powerful than proper learning. For instance, Beimel et al. \cite{BeimelKaNi10} showed that under pure differential privacy, the simple class of $\point$ functions (indicators of a single domain element) requires $\Omega(d)$ samples to privately learn properly, but only $O(\log d)$ samples to privately learn improperly. Moreover, computational separations are known between proper and improper learning even without privacy considerations. Pitt and Valiant \cite{PittVa88} showed that unless $\mathbf{NP} = \mathbf{RP}$, $k$-term DNF are not efficiently properly learnable, but they are efficiently improperly learnable \cite{Valiant84}.

Under plausible cryptographic assumptions, we resolve the question of Kasiviswanathan et al. \cite{KLNRS11} in the negative, even for improper learners. The assumption we need is the existence of ``strongly correct'' order-revealing encryption (ORE) schemes, described in Section \ref{sec:intro-ore}.

\begin{theorem}[Informal]
Assuming the existence of strongly correct ORE, there exists an efficiently computable concept class $\encthresh$ that is efficiently PAC learnable, but not efficiently learnable by any $(\eps, \delta)$-differentially private algorithm.
\end{theorem}

We stress that this result holds even for improper learners and for the relaxed notion of approximate differential privacy. We remark that cryptography has played a major role in shaping our understanding of the computational complexity of learning in a number of models (e.g. \cite{Valiant84, KearnsVa94, Kharitonov95, Servedio00}). It has also been used before to show separations between what is efficiently learnable in different models (e.g. \cite{Blum94, ServedioGo04}).

\subsection{Our Techniques}

We give an informal overview of the construction and analysis of the concept class $\encthresh$. 

We first describe the concept class of thresholds $\thresh$ and its simple PAC learning algorithm. Consider the domain $[N] = \{1, \dots, N\}$. Given a number $t \in [N]$, a threshold concept $c_t$ is defined by $c_t(x) = 1$ if and only if $x \le t$. The concept class of thresholds admits a simple and efficient proper PAC learning algorithm $L_{\thresh}$. Given a sample $\{(x_1, c_t(x_1)), \dots, (x_n, c_t(x_n))\}$ labeled by an unknown concept $c_t$, the learner $L_{\thresh}$ identifies the largest positive example $x_{i^*}$ and outputs the hypothesis $h = c_{x_{i^*}}$. That is, $L_{\thresh}$ chooses the threshold concept that minimizes the empirical error on its sample. To achieve a small constant error on \emph{any} underlying distribution on examples, it suffices to take $n = O(1)$ samples.

A simple but important observation about $L_{\thresh}$ is that it is completely oblivious to the actual numeric values of its examples, or even to the fact that the domain is $[N]$. In fact, $L_{\thresh}$ works equally well on any totally-ordered domain on which it can efficiently compare examples. In an extreme case, the learner $L_{\thresh}$ still works when its examples are encrypted under an \emph{order-revealing encryption} (ORE) scheme, which guarantees that $L_{\thresh}$ is able to learn the order of its examples, but nothing else about them. Up to small technical modifications, our concept class $\encthresh$ is exactly the class $\thresh$ where examples are encrypted under an ORE scheme.

For $\encthresh$ to be efficiently PAC learnable, it must be learnable even under distributions that place arbitrary weight on examples corresponding to invalid ciphertexts. To this end, we require a ``strong correctness'' condition on our ORE scheme. The strong correctness condition ensures that all ciphertexts, even those that are not obtained as encryptions of messages, can be compared in a consistent fashion. This condition is not met by current constructions of ORE, and one of the technical contributions of this work is a generic transformation from weakly correct ORE schemes to strongly correct ones.

While a learner similar to $L_{\thresh}$ is able to efficiently PAC learn the concept class $\encthresh$, we argue that it cannot do so while preserving differential privacy with respect to its examples. Intuitively, the security of the ORE scheme ensures that essentially the only thing a learner for $\encthresh$ can do is output a hypothesis that compares an example to one it already has. We make this intuition precise by giving an algorithm that traces the hypothesis output by any efficient learner back to one of the examples used to produce it. This formalization builds conceptually on the connection between differential privacy and traitor-tracing schemes (see Section \ref{sec:related}), but requires new ideas to adapt to the PAC learning model.

\subsection{Order-Revealing Encryption} \label{sec:intro-ore}

Motivated by the task of answering range queries on encrypted databases, an \emph{order-revealing} encryption (ORE) scheme~\cite{BCO11,BLRSZZ15} is a special type of symmetric key encryption scheme where it is possible to publicly sort \emph{ciphertexts} according to the order of the \emph{plaintexts}.  More precisely, the plaintext space of the scheme is the set of integers $[N] = \{1,...,N\}$,\footnote{More generally, any totally-ordered plaintext space can be considered} and in addition to the \emph{private} encryption and decryption procedures $\Enc,\Dec$, there is a public comparison procedure $\Comp$ that takes as input two ciphertexts, and reveals the order of the corresponding plaintexts.  The notion of \emph{best-possible semantic security}, defined in Boneh et al.~\cite{BLRSZZ15}, intuitively captures the requirement that, given a collection of ciphertexts, no information about the plaintexts is learned, \emph{except} for the ordering.

\paragraph{Known constructions of order-revealing encryption.} Order-revealing encryption can be seen as a special case of 2-input \emph{functional encryption}.  In such a scheme, there are several functions $f_1,...,f_k$, and given two ciphertexts $c_0,c_1$ encrypting $m_0,m_1$, it is possible to learn $f_i(m_0,m_1)$ for all $i \in [k]$.  General \emph{multi-input} functional encryption schemes can be obtained from indistinguishability obfuscation~\cite{GGGJKLSSZ14} or multilinear maps~\cite{BLRSZZ15}.  It is also possible to build ORE from \emph{single-input} functional encryption with function privacy, which means that $f$ is kept secret.  Such schemes can be build from regular single-input schemes without function privacy by work of Brakerski and Segev~\cite{BS15}, and such single-input schemes can also be built from obfuscation~\cite{GGHRSW13} or multilinear maps~\cite{GGHZ14b}.

\trimrm{\medskip}

Unfortunately, the above constructions are insufficient for our purposes.  The issue arises from the fact that our learner needs to work for \emph{any} distribution on ciphertexts, even distributions whose support includes malformed ciphertexts.  Unfortunately, previous constructions only achieve a weak form of correctness, which guarantees that encrypting two messages and then comparing the ciphertexts using $\Comp$ produces the same result (with overwhelming probability) as comparing the plaintexts directly.  This requirement only specifies how $\Comp$ works on \emph{valid} ciphertexts, namely actual encryptions of messages.  Moreover, correctness is only guaranteed for these messages with overwhelming probability, meaning even some valid ciphertexts may cause $\Comp$ to misbehave.  

For our learner, this weak form of correctness means, for some distributions that place significant weight on bad ciphertexts, the comparison procedure is completely useless, and thus the learner will fail for these distributions.  

\trimrm{\medskip}

We therefore need a stronger correctness guarantee.  We need that, for any two \emph{ciphertexts}, the comparison procedure is consistent with decrypting the two ciphertexts and comparing the resulting plaintexts.  This correctness guarantee is meaningful even for improperly generated ciphertexts.

We note that none of the existing constructions of order-revealing encryption outlined above satisfy this stronger notion.  For the obfuscation-based schemes, ciphertexts consist of obfuscated programs.  In these schemes, it is easy to describe invalid ciphertexts where the obfuscated program performs incorrectly, causing the comparison procedure to output the wrong result.  In the multilinear map-based schemes, the underlying instantiation use current ``noisy''  multilinear maps, such as~\cite{GGH13}.  An invalid ciphertext could, for example, have too much noise, which will cause the comparison procedure to behave unpredictably.

\trimtext{\vspace{-2mm}}

\paragraph{Obtaining strong correctness.}  

We first argue that, for all existing ORE schemes, the scheme can be modified so that $\Comp$ is correct for all \emph{valid} ciphertexts. We then give a generic conversion from any ORE scheme with weakly correct comparison, including the tweaked existing schemes, into a strongly correct scheme.  We simply modify the ciphertext by adding a non-interactive zero-knowledge (NIZK) proof that the ciphertext is well-formed, with the common reference string added to the public comparison key.  Then the decryption and comparison procedures check the proof(s), and only output the result (either decryption or comparison) if the proof(s) are valid.  The (computational) zero-knowledge property of the NIZK implies that the addition of the proof to the ciphertext does not affect security.  Meanwhile, NIZK soundness implies that any ciphertext accepted by the decryption and comparison procedures must be valid, and the weak correctness property of the underlying ORE implies that for valid ciphertexts, decryption and comparison are consistent.  The result is that comparisons are consistent with decryption \emph{for all} ciphertexts, giving strong correctness.

As we need strong correctness for every ciphertext, even hard-to-generate ones, we need the NIZK proofs to have perfect soundness, as opposed to computational soundness.  Such NIZK proofs were built in~\cite{GOS12}.

We note also that the conversion outlined above is not specific to ORE, and applies more generally to functional encryption schemes.

\subsection{Related Work} \label{sec:related}

\paragraph{Hardness of Private Query Release.}
One of the most basic and well-studied statistical tasks in differential privacy is the problem of releasing answers to \emph{counting queries}. A counting query asks,``what fraction of the records in a dataset $D$ satisfy the predicate $q$?''. Given a collection of $k$ counting queries $q_1, \dots, q_k$ from a family $\cQ$, the goal of a query release algorithm is to release approximate answers to these queries while preserving differential privacy. A remarkable result of Blum et al. \cite{BlumLiRo08}, with subsequent improvements by \cite{DNRRV09, DworkRoVa10, RothRo10, HardtRo10, GuptaRoUl12, HardtLiMc12}, showed that an arbitrary sequence of counting queries can be answered accurately with differential privacy even when $k$ is exponential in the dataset size $n$. Unfortunately, all of these algorithms that are capable of answering more than $n^2$ queries are inefficient, running in time exponential in the dimensionality of the data. Moreover, several works \cite{DNRRV09, Ullman13, BonehZh14} have gone on to show that this inefficiency is likely inherent.

These computational lower bounds for private query release rely on a connection between the hardness of private query release and \emph{traitor-tracing schemes}, which was first observed by Dwork et al. \cite{DNRRV09}. Traitor-tracing schemes were introduced by Chor, Fiat, and Naor \cite{ChorFiNa94} to help digital content producers identify pirates as they illegally redistribute content. Traitor-tracing schemes are conceptually analogous to the example reidentification scheme we use to obtain our hardness result for private learning. Instantiating this connection with the traitor-tracing scheme of Boneh, Sahai, and Waters \cite{BSW06}, which relies on certain assumptions in bilinear groups, Dwork et al. \cite{DNRRV09} exhibited a family of $2^{\tilde{O}(\sqrt{n})}$ queries for which no efficient algorithm can produce a data structure which could be used to answer all queries in this family. Very recently, Boneh and Zhandry \cite{BonehZh14} constructed a new traitor-tracing scheme based on indistinguishability obfuscation that yields the same infeasibility result for a family of $n\cdot 2^{O(d)}$ queries on records of size $d$. Extending this connection, Ullman \cite{Ullman13} constructed a specialized traitor-tracing scheme to show that no efficient private algorithm can answer more than $\tilde{O}(n^2)$ arbitrary queries that are given as input to the algorithm.

\trimrm{
Dwork et al. \cite{DNRRV09} also showed strong lower bounds against private algorithms for producing \emph{synthetic data}. Synthetic data generation algorithms produce a new ``fake'' dataset, whose rows are of the same type as those in the original dataset, with the promise that the answers to some restricted set of queries on the synthetic dataset well-approximate the answers on the original dataset. Assuming the existence of one-way functions, Dwork et al. \cite{DNRRV09} exhibited an efficiently computable collection of queries for which no efficient private algorithm can produce useful synthetic data. Ullman and Vadhan \cite{UllmanVa11} refined this result to hold even for extremely simple classes of queries.

Nevertheless, the restriction to synthetic data is significant to these results, and they do not rule out the possibility that other privacy-preserving data structures can be used to answer large families of restricted queries. In fact, when the synthetic data restriction is lifted, there are algorithms (e.g. \cite{HardtRoSe12, ThalerUlVa12,  ChandrasekaranThUlWa14, DworkNiTa14}) that answer queries from certain exponentially large families in subexponential time. One can view the problem of synthetic data generation as analogous to proper learning. In both cases, placing natural syntactic restrictions on the output of an algorithm may in fact come at the expense of utility or computational efficiency.
}
\trimtext{\vspace{-2mm}}

\paragraph{Efficiency of SQ Learning.}
Feldman and Kanade \cite{FeldmanKa12} addressed the question of whether information-theoretically efficient SQ learners -- i.e., those making polynomially many queries -- could be made computationally efficient. One of their main negative results showed that unless $\mathbf{NP} = \mathbf{RP}$, there exists a concept class with polynomial query complexity that is not efficiently SQ learnable. Moreover, this concept class is efficiently PAC learnable, which suggests that the restriction to SQ learning can introduce an inherent computational cost.

We show that the concept class $\encthresh$ can be learned (inefficiently) with polynomially many statistical queries. The result of Blum et al. \cite{BlumDwMcNi05} discussed above, showing that SQ learning algorithms can be efficiently simulated by differentially private algorithms, thus shows that $\encthresh$ also separates SQ learners making polynomially many queries from computationally efficient SQ learners.

\begin{corollary}[Informal]
Assuming the existence of strongly correct ORE, the concept class $\encthresh$ is efficiently PAC learnable and has polynomial SQ query complexity, but is not efficiently SQ learnable.
\end{corollary}

While our proof relies on much stronger hardness assumptions, it reveals ORE as a new barrier to efficient SQ learning. As discussed in more detail in \trimrm{Section \ref{sec:sq}}\trimtext{the full version of this work}, even though their result is about computational hardness, Feldman and Kanade's choice of a concept class relies crucially on the fact that parities are hard to learn in the SQ model even information-theoretically. By contrast, our concept class $\encthresh$ is computationally hard to SQ learn for a reason that appears fundamentally different than the information-theoretic hardness of SQ learning parities.

\trimtext{\vspace{-2mm}}

\paragraph{Learning from Encrypted Data.}
Several works have developed schemes for training, testing, and classifying machine learning models over encrypted data (e.g. \cite{GraepelLaNa12, BPTG14}). In a model use case, a client holds a sensitive dataset, and uploads an encrypted version of the dataset to a cloud computing service. The cloud service then trains a model over the encrypted data and produces an encrypted classifier it can send back to the client, ideally without learning anything about the examples it received. The notion of privacy afforded to the individuals in the dataset here is complementary to differential privacy. While the cloud service does not learn anything about the individuals in the dataset, its output might still depend heavily on the data of certain individuals.

In fact, our non-differentially private PAC learner for the class $\encthresh$ exactly performs the task of learning over encrypted data, producing a classifier without learning anything about its examples beyond their order (this addresses the difficulty of implementing comparisons from prior work \cite{GraepelLaNa12}). Thus one can interpret our results as showing that not only are these two notions of privacy for machine learning training complementary, but that they may actually be in conflict. Moreover, the strong correctness guarantee we provide for ORE (which applies more generally to multi-input functional encryption) may help enable the theoretical study of learning from encrypted data in other PAC-style settings.

\section{Preliminaries and Definitions}

\label{sec:prelim}

\trimrm{
\subsection{PAC Learning and Private PAC Learning}

For each $k \in \N$, let $X_k$ be an instance space (such as $\{0, 1\}^k$), where the parameter $k$ represents the size of the elements in $X_k$. Let $\cC_k$ be a set of boolean functions $\{c : X_k \to \{0, 1\}\}$. The sequence $(X_1, \cC_1), (X_2, \cC_2), \dots$ represents an infinite sequence of learning problems defined over instance spaces of increasing dimension. We will generally suppress the parameter $k$, and refer to the problem of learning $\cC$ as the problem of learning $\cC_k$ for every $k$.

A learner $L$ is given examples sampled from an unknown probability distribution $\cD$ over $X$, where the examples are labeled according to an unknown {\em target concept} $c\in \cC$. The learner must select a hypothesis $h$ from a hypothesis class $\cH$ that approximates the target concept with respect to the distribution $\cD$. More precisely,

\begin{definition}
The {\em generalization error} of a hypothesis $h:X\rightarrow\{0,1\}$ (with respect to a target concept $c$ and distribution $\cD$) is defined by
$\error_{\cD}(c,h)=\Pr_{x \sim \cD}[h(x)\neq c(x)].$
If $\error_{\cD}(c,h)\leq\alpha$ we say that $h$ is an {\em $\alpha$-good} hypothesis for $c$ on $\cD$.
\end{definition}

\begin{definition}[PAC Learning~\cite{Valiant84}]\label{def:PAC}
Algorithm $L : (X \times \{0, 1\})^n \to \cH$ is an {\em $(\alpha,\beta)$-accurate PAC learner} for the concept class $\cC$ using hypothesis class $\cH$ with sample complexity $n$ if for all target concepts $c \in \cC$ and all distributions $\cD$ on $X$, given an input of $n$ samples $S =((x_i, c(x_i)),\ldots,(x_n, c(x_n)))$, where each $x_i$ is drawn i.i.d.\ from $\cD$, algorithm $L$ outputs a hypothesis $h\in \cH$ satisfying
$\Pr[\error_{\cD}(c,h)  \leq \alpha] \geq 1-\beta$. The probability here is taken over the random choice of the examples in $S$ and the coin tosses of the learner $L$.

The learner $L$ is \emph{efficient} if it runs in time polynomial in the size parameter $k$, the representation size of the target concept $c$, and the accuracy parameters $1/\alpha$ and $1/\beta$. Note that a necessary (but not sufficient) condition for $L$ to be efficient is that its sample complexity $n$ is polynomial in the learning parameters.

 If $\cH\subseteq \cC$ then $L$ is called a {\em proper} learner. Otherwise, it is called an {\em improper} learner.
\end{definition}

Kasiviswanathan et al. \cite{KLNRS11} defined a \emph{private learner} as a PAC learner that is also differentially private. Recall the definition of differential privacy:
\begin{definition}
A learner $L : (X \times \{0, 1\})^n \to \cH$ is  $(\eps, \delta)$-\emph{differentially private} if for all sets $T \subseteq \cH$, and neighboring sets of examples $S \sim S'$,
\[\Pr[L(S) \in T] \le e^{\eps}\Pr[L(S') \in T] + \delta.\]
\end{definition}

}

The technical object that we will use to show our hardness results for differential privacy is what we call an \emph{example reidentification scheme}. It is analogous to the hard-to-sanitize database distributions \cite{DNRRV09, UllmanVa11} and re-identifiable database distributions \cite{BUV14} used in prior works to prove hardness results for private query release, but is adapted to the setting of computational learning. In the first step, an algorithm $\GenTT$ chooses a concept and a sample $S$ labeled according to that concept. In the second step, a learner $L$ receives either the sample $S$ or the sample $S_{-i}$ where an appropriately chosen example $i$ is replaced by a junk example, and learns a hypothesis $h$. Finally, an algorithm $\TraceTT$ attempts to use $h$ to identify one of the rows given to $L$. If $\TraceTT$ succeeds at identifying such a row with high probability, then it must be able to distinguish $L(S)$ from $L(S_{-i})$, showing that $L$ cannot be differentially private. We formalize these ideas below.

\begin{definition}
An $(\alpha, \xi)$-\emph{example reidentification scheme} for a concept class $\cC$ consists of a pair of algorithms, $(\GenTT, \TraceTT)$ with the following properties.
\begin{description}
\item[$\GenTT(k, n)$] Samples a concept $c \in \cC_k$ and an associated distribution $\cD$. Draws i.i.d. examples $x_1, \dots, x_n \getsr \cD$, and a fixed value $x_0$. Let $S$ denote the labeled sample $((x_1, c(x_1)), \dots, (x_n, c(x_n))$, and for any index $i \in [n]$, let $S_{-i}$ denote the sample with the pair $(x_i, c(x_i))$ replaced with $(x_0, c(x_0))$.
\item[$\TraceTT(h)$] Takes state shared with $\GenTT$ as well as a hypothesis $h$ and identifies an index in $[n]$ (or $\bot$ if none is found).
\end{description}
The scheme obeys the following ``completeness'' and ``soundness'' criteria on the ability of $\TraceTT$ to identify an example given to a learner $L$.
\paragraph{Completeness.} A good hypothesis can be traced to some example. That is, for every efficient learner $L$,
\[\Pr[\error_{\cD}(c, h) \le \alpha \land \TraceTT(h) = \bot] \le \xi.\]
Here, the probability is taken over $h\getsr L(S)$ and the coins of $\GenTT$ and $\TraceTT$. 
\paragraph{Soundness.} For every efficient learner $L$, $\TraceTT$ cannot trace $i$ from the sample $S_{-i}$. That is, for all $i \in [n]$,
\[\Pr[\TraceTT(h) = i] \le \xi\]
for $h \getsr L(S_{-i})$.
\end{definition}

\trimrm{
We may sometimes relax the completeness condition to hold only under certain restrictions on $L$'s output (e.g. $L$ is a proper learner or $L$ is a representation learner). In this case, we say the $(\GenTT, \TraceTT)$ is an example reidentification scheme for (properly, representation) learning a class $\cC$.
}

\begin{theorem}\label{thm:diffpac}
Let $(\GenTT, \TraceTT)$ be an $(\alpha, \xi)$-example reidentification scheme for a concept class $\cC$. Then for every $\beta>0$ and polynomial $n(k)$, there is no efficient $(\eps, \delta)$-differentially private $(\alpha, \beta)$-PAC learner for $\cC$ using $n$ samples when
\[\delta < \left(\frac{1 - \beta - \xi}{n}\right) - e^{\eps} \xi.\]
\end{theorem}
In a typical setting of parameters, we will take $\alpha, \beta, \eps = O(1)$ and $\delta, \xi = o(1/n)$, in which case the inequality in Theorem~\ref{thm:diffpac} will be satisfied for sufficiently large $n$.

\trimrm{
\begin{proof}
Suppose instead that there were a computationally efficient $(\eps, \delta)$-differentially private $(\alpha, \beta)$-PAC learner $L$ for $\cC$ using $n$ samples. Then there exists an $i \in [n]$ such that $\Pr[\TraceTT(L(S)) = i] \ge (1-\beta -\xi) / n$. However, since $L$ is differentially private, \[\Pr[\TraceTT(L(S_{-i})) = i] \ge e^{-\eps}\left(\frac{1-\beta-\xi}{n} - \delta\right) > \xi(n),\] which contradicts the soundness of $(\GenTT, \TraceTT)$.  
\end{proof}
}

\subsection{Order-Revealing Encryption}

\begin{definition} An Order-Revealing Encryption (ORE) scheme is a tuple $(\Gen,\Enc,\Dec,\Comp)$ of algorithms where:
	\begin{itemize}
		\item $\Gen(1^\lambda,1^\ell)$ is a randomized procedure that takes as inputs a security parameter $\lambda$ and plaintext length $\ell$, and outputs a secret encryption/decryption key $\sk$ and public parameters $\params$.
		\item $\Enc(\sk,m)$ is a potentially randomized procedure that takes as input a secret key $\sk$ and a message $m\in\{0,1\}^\ell$, and outputs a ciphertext $c$.
		\item $\Dec(\sk,c)$ is a deterministic procedure that takes as input a secret key $\sk$ and a ciphertext $c$, and outputs a plaintext message $m\in\{0,1\}^\ell$ or a special symbol $\bot$.
		\item $\Comp(\params,c_0,c_1)$ is a deterministic procedure that ``compares'' two ciphertexts, outputting either ``$>$'', ``$<$'', ``$=$'', or $\bot$.
	\end{itemize}
\end{definition}

\paragraph{Correctness.} An ORE scheme must satisfy two separate correctness requirements:
\begin{itemize}
	\item {\bf Correct Decryption:} This is the standard notion of correctness for an encryption scheme, which says that decryption succeeds.  We will only consider \emph{strongly} correct decryption, which requires that decryption \emph{always} succeeds.  For all security parameters $\lambda$ and message lengths $\ell$,
		\[\Pr[\Dec(\sk,\;\Enc(\sk,m)\;) = m: (\sk,\params)\gets\Gen(1^\lambda,1^\ell)] = 1.\]

	\trimrm{
	
	\item {\bf Correct Comparison:} We require that the comparison function succeeds.  We will consider two notions, namely \emph{strong} and \emph{weak}. In order to define these notions, we first define two auxiliary functions:
	\begin{itemize}
		\item $\Comp_{plain}(m_0,m_1)$ is just the plaintext comparison function. That is, for $m_0<m_1$, $\Comp_{plain}(m_0,m_1)=``<"$, $\Comp_{plain}(m_1,m_0)=``>"$, and $\Comp_{plain}(m_0,m_0)=``="$.
		\item $\Comp_{ciph}(\sk,c_0,c_1)$ is a ciphertext comparison function which uses the secret key.  If first computes $m_b=\Dec(\sk,c_b)$ for $b=0,1$.  If either $m_0=\bot$ or $m_1=\bot$ (in other words, if either decryption failed), then $\Comp_{ciph}$ outputs $\bot$.  If both $m_0,m_1\neq\bot$, then the output is $\Comp_{plain}(m_0,m_1)$.
	\end{itemize}
	
	Now we define our comparison correctness notions:
	
	\begin{itemize}
		\item {\bf Weakly Correct Comparison:} This informally requires that comparison is consistent with encryption.  For all security parameters $\lambda$, message lengths $\ell$, and messages $m_0,m_1\in\{0,1\}^\ell$, 
			\[\Pr\left[\Comp(\params,c_0,c_1)=\Comp_{plain}(m_0,m_1):\begin{array}{c}(\sk,\params)\gets\Gen(1^\lambda,1^\ell)\\ c_b\gets\Enc(\sk,m_b)\end{array}\right]=1.\]
		In particular, for correctly generated ciphertexts, $\Comp$ never outputs $\bot$.
		
		\item {\bf Strongly Correct Comparison:} This informally requires that comparison is consistent with \emph{decryption}.  For all security parameters $\lambda$, message lengths $\ell$, and ciphertexts $c_0,c_1$, 
			\[\Pr\left[\Comp(\params,c_0,c_1)=\Comp_{ciph}(\sk,c_0,c_1):(\sk,\params)\gets\Gen(1^\lambda,1^\ell)\right]=1.\]
		
	\end{itemize}
	
	}
	
	\trimtext{
	
	\item {\bf Strongly Correct Comparison:} This informally requires that comparison is consistent with \emph{decryption}.  For all security parameters $\lambda$, message lengths $\ell$, and ciphertexts $c_0,c_1$, 
			\[\Pr\left[\Comp(\params,c_0,c_1)=\Comp_{ciph}(\sk,c_0,c_1):(\sk,\params)\gets\Gen(1^\lambda,1^\ell)\right]=1.\]
	Here, the function $\Comp_{ciph}$ computes $m_b=\Dec(\sk,c_b)$ for $b=0,1$.  If either $m_0=\bot$ or $m_1=\bot$ (in other words, if either decryption failed), then $\Comp_{ciph}$ outputs $\bot$.  If both $m_0,m_1\neq\bot$, then the output is $``<", ``>"$, or $``="$ if $m_0 < m_1$, $m_0 > m_1$, or $m_0 = m_1$, respectively.
			
	}
	
\end{itemize} 

\paragraph{Security.} For security, we will consider a relaxation of the ``best possible'' security notion of Boneh et al.~\cite{BLRSZZ15}.  Namely, we only consider static adversaries that submit all queries at once.  ``Best possible'' security is a modification of the standard notion of CPA security for symmetric key encryption to block trivial attacks.  That is, since the comparison function always leaks the order of the plaintexts, the left and right sets of challenge messages must have the same order.  In our relaxation where all challenge messages are queried at once, we can therefore assume without loss of generality that the left and right sequences of messages are sorted in ascending order.  For simplicity, we will also disallow the adversary from querying on the same message more than once.  This gives the following definition:

\begin{definition} An ORE scheme $(\Gen,\Enc,\Dec,\Comp)$ is \emph{statically secure} if, for all efficient adversaries $\cA$,  $|\Pr[W_0]-\Pr[W_1]|$ is negligible, where $W_b$ is the event that $\cA$ outputs $1$ in the following experiment:
	\begin{itemize}
		\item $\cA$ produces two message sequences $m_1^{(L)}<m_2^{(L)}<\dots<m_q^{(L)}$ and $m_1^{(R)}<m_2^{(R)}<\dots<m_q^{(R)}$
		\item The challenger runs $(\sk,\params)\gets\Gen(1^\lambda,1^\ell)$.  It then responds to $\cA$ with $\params$, as well as $c_1,\dots,c_q$ where 
		\[c_i=\begin{cases}
			\Enc(\sk,m_i^{(L)})&\text{if }b=0\\
			\Enc(\sk,m_i^{(R)})&\text{if }b=1
		\end{cases}\]
		\item $\cA$ outputs a guess $b'$ for $b$.
	\end{itemize}
\end{definition}

\trimrm{

We also consider a weaker definition, which only allows the sequences $m_i^{(L)}$ and $m_i^{(R)}$ to differ at a single point:

\begin{definition} An ORE scheme $(\Gen,\Enc,\Dec,\Comp)$ is \emph{statically single-challenge secure} if, for all efficient adversaries $\cA$,  $|\Pr[W_0]-\Pr[W_1]|$ is negligible, where $W_b$ is the event that $\cA$ outputs $1$ in the following experiment:
	\begin{itemize}
		\item $\cA$ produces a sequence of messages $m_1<m_2<\dots<m_q$, and challenge messages $m_L,m_R$ such that $m_i<m_L<m_R<m_{i+1}$ for some $i\in[q-1]$.
		\item The challenger runs $(\sk,\params)\gets\Gen(1^\lambda,1^\ell)$.  It then responds to $\cA$ with $\params$, as well as $c_1,\dots,c_q$ where $c_i=\Enc(\sk,m_i)$ and 
		\[c^*=\begin{cases}
			\Enc(\sk,m_L)&\text{if }b=0\\
			\Enc(\sk,m_R)&\text{if }b=1
			\end{cases}\]
		\item $\cA$ outputs a guess $b'$ for $b$.
	\end{itemize}
\end{definition}

We now argue that these two definitions are equivalent up to some polynomial loss in security.  

\begin{theorem}\label{thm:singlemessage}$(\Gen,\Enc,\Dec,\Comp)$ is statically secure if and only if it is statically single-challenge secure.
\end{theorem}

\begin{proof} We prove that single-challenge security implies many-challenge security through a sequence of hybrids.  Each hybrid will only differ in the messages $m_i$ that are encrypted, and each adjacent hybrid will only differ in a single message.  The first hybrid will encrypt $m_i^{(L)}$, and the last hybrid will encrypt $m_i^{(R)}$.  Thus, by applying the single-challenge security for each hybrid, we conclude that the first and last hybrids are indistinguishable, thus showing many-challenge security.
	
\paragraph{Hybrid $j$ for $j\leq q$.} 
\[m_i=\begin{cases}
	\min(m_i^{(L)},m_i^{(R)})&\text{if }i\leq j\\
	m_i^{(L)}&\text{if }i>j
\end{cases}\]
First, notice that all the $m_i$ are in order since both sequences $m_i^{(L)}$ and $m_i^{(R)}$ are in order.  Second, the only difference between {\bf Hybrid $(j-1)$} and {\bf Hybrid $j$} is that $m_j=m_j^{(L)}$ in {\bf Hybrid $(j-1)$} and $m_j=\min(m_j^{(L)},m_j^{(R)})$ in {\bf Hybrid $j$}.  Thus, single-challenge security implies that each adjacent hybrid is indistinguishable.  Moreover, for $j$ where $m_j^{(L)}<m_j^{(R)}$, the two hybrids are actually identical.

\paragraph{Hybrid $j$ for $j> q$.} 
\[m_i=\begin{cases}
\min(m_i^{(L)},m_i^{(R)})&\text{if }i\leq 2q-j\\
m_i^{(R)}&\text{if }i>2q-j
\end{cases}\]
Again, notice that all the $m_i$ are in order.  Moreover, the only different between {\bf Hybrid $(2q-j)$} and {\bf Hybrid $(2q-j+1)$} is that $m_{j}=\min(m_j^{(L)},m_j^{(R)})$ in {\bf Hybrid $(2q-j)$} and $m_j=m_j^{(R)}$ in {\bf Hybrid $(2q-j+1)$}.  Thus, single-challenge security implies that each adjacent hybrid is indistinguishable.  Moreover, for $j$ where $m_j^{(L)}>m_j^{(R)}$, the two hybrids are actually identical.
	
\end{proof}

}

\section{The Concept Class $\encthresh$ and its Learnability}

Let $(\Gen, \Enc, \Dec, \Comp)$ be a statically secure ORE scheme with strongly correct comparison. We define a concept class $\encthresh$, which intuitively captures the class of threshold functions where examples are encrypted under the ORE scheme. Throughout this discussion, we will take $N = 2^\ell$ and regard the plaintext space of the ORE scheme to be $[N] = \{1, \dots, N\}$. Ideally we would like, for each threshold $t \in [N+1]$ and each $(\sk, \params) \gets \Gen(1^\lambda)$, to define a concept 
\[f_{t, \sk, \params}(c) = \begin{cases}
1 & \text{ if } \Dec_{\sk}(c) < t \\
0 & \text{ otherwise.}
\end{cases}\]
However, we need to make a few technical modifications to ensure that $\encthresh$ is efficiently PAC learnable.
\begin{enumerate}
\item In order for the learner to be able to use the comparison function $\Comp$, it must be given the public parameters $\params$ generated by the ORE scheme. We address this in the natural way by attaching a set of public parameters to each example. Moreover, we define $\encthresh$ so that each concept is supported on the single set of public parameters that corresponds to the secret key used for encryption and decryption.
\item Only a subset of binary strings form valid $(\sk, \params)$ pairs that are actually produced by $\Gen$ in the ORE scheme. To represent concepts, we need a reasonable encoding scheme for these valid pairs. The encoding scheme we choose is the polynomial-length sequence of random coin tosses used by the algorithm $\Gen$ to produce $(\sk, \params)$.
\end{enumerate}

We now formally describe the concept class $\encthresh$. Each concept is parameterized by a string $r$, representing the coin tosses of the algorithm $\Gen$, and a threshold $t \in [N+1]$ for $N = 2^\ell$. In what follows, let $(\sk^r, \params^r)$ be the keys output by $\Gen(1^\lambda, 1^\ell)$ when run on the sequence of coin tosses $r$. Let
\[f_{t, r}(\params, c) = \begin{cases}
1 & \text{if } (\params = \params^r) \land (\Dec(\sk^r, c) \ne \bot) \land (\Dec(\sk^r, c) < t) \\
0 & \text{otherwise.}
\end{cases}\]
Notice that given $t$ and $r$, the concept $f_{t, r}$ can be efficiently evaluated. The description length $k$ of the instance space $X_k = \{0, 1\}^k$ is polynomial in the security parameter $\lambda$ and plaintext length $\ell$.

\subsection{An Efficient PAC Learner for $\encthresh$}

\trimtext{In the full version of this work, we}\trimrm{We} argue that $\encthresh$ is efficiently PAC learnable by formalizing the argument from the introduction. Because we need to include the ORE public parameters in each example, the PAC learner $L$ \trimrm{(Algorithm \ref{fig:pac-learner})} for $\encthresh$ actually works in two stages. In the first stage, $L$ determines whether there is significant probability mass on examples corresponding to some public parameters $\params$. Recall that each concept in $\encthresh$ is supported on exactly one such set of parameters. If there is no significant mass on any $\params$, then the all-zeroes hypothesis is a good hypothesis. On the other hand, if there is a heavy set of parameters, the learner $L$ applies $\Comp$ using those parameters to learn a good comparator.

\begin{theorem}
Let $\alpha, \beta > 0$. There exists a PAC learning algorithm $L$ for the concept class $\encthresh$ achieving error $\alpha$ and confidence $1-\beta$. Moreover, $L$ is efficient (running in time polynomial in the parameters $k, 1/\alpha, \log(1/\beta)$).
\end{theorem}

\trimrm{

\begin{algorithm}[h]
\caption{Learner $L$ for $\encthresh$}
\label{fig:pac-learner}
\begin{enumerate}
\item Request examples $\{(\params_1, c_1, b_1), \dots, (\params_n, c_n, b_n)\}$ for $n = \lceil\log(1/\beta)/\alpha\rceil$.
\item Identify an $i$ for which $b_i = 1$ and set $\params^* = \params_i$; if no such $i$ exists, return $h \equiv 0$.
\item Let $G = \{j : \params_j = \params^*,b_j = 1\}$. Let $j^* \in G$ be an index with $\Comp(\params^*, c_j, c_{j^*}) \in \{<, =, \bot\}$ for all $j \in G$.
\item Return $h$ defined by
\[h(\params, c) = \begin{cases}
1 \quad \text{if } (\params = \params^*) \land (\Comp(\params^*, c, c_{j^*}) \in \{<, =\}) \\
0 \quad \text{otherwise.}
\end{cases}\]
\end{enumerate}
\end{algorithm}

\begin{proof}

Fix a target concept $f_{t, r} \in \encthresh_k$ and a distribution $\cD$ on examples. First observe that the learner $L$ always outputs a hypothesis with one-sided error, i.e. we always have $h \le f_{t, r}$ pointwise. Also observe that $f_{t',r}\leq f_{t,r}$ pointwise for any $t'<t$.  These both follow from the strong correctness of the ORE scheme. Let $(\sk^r, \params^r)$ denote the keys output by $\Gen(1^\lambda, 1^\ell)$ when run on the sequence of coin tosses $r$. Let $\pos$ denote the set of examples $(\params, c)$ on which $f_{t, r}(\params, c) = 1$. We divide the analysis of the learner in to two cases based on the weight $\cD$ places on $\pos$.

\paragraph{Case 1:} $\cD$ places weight at least $\alpha$ on $\pos$. Define $\hat{t} \in [N+1]$ as the largest $\hat{t} \le t$ such that $\error_{\cD}(f_{\hat{t}, r}, f_{t, r}) \ge \alpha$.  Such a $\hat{t}$ is guaranteed to exist since $f_{0,r}$ is the all-zeros function, and therefore $\error_{\cD}(f_{0,r},f_{t,r})$ is equal to the weight $\cD$ places on $\pos$, which is at least $\alpha$.  

Suppose $f_{\hat{t}+1,r}\leq h$ pointwise.  Since $h$ has one-sided error (that is, $h\leq f_{t,r}$ pointwise), we have $\error_{\cD}(f_{\hat{t}+1,r},f_{t,r})=\error_{\cD}(f_{\hat{t}+1,r},h)+\error_{\cD}(h,f_{t,r})$, or \[\error_{\cD}(h,f_{t,r})=\error_{\cD}(f_{\hat{t}+1,r},f_{t,r})-\error_{\cD}(f_{\hat{t}+1,r},h)\le \error_{\cD}(f_{\hat{t}+1,r},f_{t,r}) <\alpha.\]

Therefore, it suffices to show that $f_{\hat{t}+1,r}\leq h$ with probability at least $1-\beta$. This is guaranteed as long as $L$ receives a sample $(\params^r, c_i, 1)$ with $\hat{t} \le \Dec(\sk^r, c_i) < t$.  In other words, $f_{t,r}(\params^r,c_i)=1$ and $f_{\hat{t},r}(\params^r,c_i)=0$.  Since $f_{\hat{t},r}\le f_{t,r}$ pointwise, such samples exactly account for the error between $f_{\hat{t},r}$ and $f_{t,r}$.  Thus since $\error_{\cD}(f_{\hat{t}, r}, f_{t, r}) \ge \alpha$, for each $i$ it must be that $\hat{t} \le \Dec(\sk^r, c_i) < t$ with probability at least $\alpha$.  The learner $L$ therefore receives \emph{some} sample $c_i$ with $\hat{t} \le \Dec(\sk^r, c_i) < t$ with probability at least  $1- (1-\alpha)^n \ge 1-\beta$ (since we took $n \ge \log(1/\beta)/\alpha$).

\paragraph{Case 2:} $\cD$ places less than $\alpha$ weight on $\pos$. Then the identically zero hypothesis has error at most $\alpha$, so the claim holds because $0 \le h \le f_{t, r}$.

\end{proof}

}

\subsection{Hardness of Privately Learning $\encthresh$}

We now \trimrm{prove}\trimtext{sketch} the hardness of privately learning $\encthresh$ by constructing an example reidentification scheme for this concept class. Recall that an example reidentification scheme consists of two algorithms, $\GenTT$, which selects a distribution, a concept, and examples to give to a learner, and $\TraceTT$ which attempts to identify one of the examples the learner received. 

Our example reidentification scheme yields a hard distribution even for \emph{weak-learning}, where the error parameter $\alpha$ is taken to be inverse-polynomially close to $1/2$.

\begin{theorem} \label{thm:reid-main}
Let $\gamma(n)$ and $\xi(n)$ be noticeable functions. Let $(\Gen, \Enc, \Dec, \Comp)$ be a statically single-challenge secure ORE scheme. Then there exists an (efficient) $(\alpha = \frac{1}{2} - \gamma, \xi)$-example reidentification scheme $(\GenTT, \TraceTT)$ for the concept class $\encthresh$.
\end{theorem}

We \trimrm{start with}\trimtext{give} an informal description of the scheme $(\GenTT, \TraceTT)$. The algorithm $\GenTT$ sets up the parameters of the ORE scheme, chooses the ``middle'' threshold concept corresponding to $t = N/2$, and sets the distribution on examples to be encryptions of uniformly random messages (together with the correct public parameters needed for comparison). Let $m_1 < m_2 < \dots < m_n$ denote the sorted sequence of messages whose encryptions make up the sample produced by $\GenTT$ (with overwhelming probability, they are indeed distinct). We can thus break the plaintext space up into buckets of the form $B_i = [m_i, m_{i+1})$. Suppose $L$ is a (weak) learner that produces a hypothesis $h$ with advantage $\gamma$ over random guessing. Such a hypothesis $h$ must be able to distinguish encryptions of messages $m \le t$ from encryptions of messages $m > t$ with advantage $\gamma$. Thus, there must be a pair of adjacent buckets $B_{i-1}, B_i$ for which $h$ can distinguish encryptions of messages from $B_{i-1}$ from encryptions from $B_i$ with advantage $\frac{\gamma}{n}$.

This observation leads to a natural definition for $\TraceTT$: locate a pair of adjacent buckets $B_{i-1}, B_i$ that $h$ distinguishes, and output the identity $i$ of the example separating those buckets. Completeness of the resulting scheme, i.e. the fact that some example is reidentified when $L$ succeeds, follows immediately from the preceding discussion. We argue soundness, i.e. that an example absent from $L$'s sample is not identified, by reducing to the static security of the ORE scheme. The intuition is that if $L$ is not given example $i$, then it should not be able to distinguish encryptions from bucket $B_{i-1}$ from encryptions from bucket $B_{i}$.

To make the security reduction somewhat more precise, suppose for the sake of contradiction that there is an efficient algorithm $L$ that violates the soundness of $(\GenTT, \TraceTT)$ with noticeable probability $\xi$. That is, there is some $i$ such that even without example $i$, the algorithm $L$ manages to produce (with probability $\xi$) a hypothesis $h$ that distinguishes $B_{i-1}$ from $B_{i}$. A natural first attempt to violate the security of the ORE is to construct an adversary that challenges on the message sequences $m_1 <\dots < m_{i-1} < m_i^{(L)} < m_{i+1}, < , m_n$ and $m_1 < \dots < m_{i-1} < m_i^{(R)} < m_{i+1} < \dots < m_n$, where $m_i^{(L)}$ is randomly chosen from $B_{i-1}$ and $m_i^{(R)}$ is randomly chosen from $B_{i}$. Then if $h$ can distinguish $B_{i-1}$ from $B_i$, the adversary can distinguish the two sequences. Unfortunately, this approach fails for a somewhat subtle reason. The hypothesis $h$ is only guaranteed to distinguish $B_{i-1}$ from $B_{i}$ \emph{with probability $\xi$}. If $h$ fails to distinguish the buckets -- or distinguishes them in the opposite direction -- then the adversary's advantage is lost.

To overcome this issue, we instead rely on the security of the ORE for sequences that differ on \emph{two} messages.  For the ``left'' challenge, our adversary samples two messages from the same randomly chosen bucket, $B_{i-1}$ or $B_{i}$ (in addition to requesting encryptions of $m_1, \dots, m_{i-1}, m_i, \dots, m_n$). For the ``right'' challenge, it samples one message from each bucket $B_{i-1}$ and $B_i$. Let $c^0$ and $c^1$ be the ciphertexts corresponding to thee challenge messages. If $h$ agrees on $c^0$ and $c^1$, then this suggests the messages are from the same bucket, and the adversary should guess ``left''. On the other hand, if $h$ disagrees on $c^0$ and $c^1$, then the adversary should guess ``right''. If $h$ distinguishes the buckets $B_{i-1}$ and $B_i$, this adversary does strictly better than random guessing. On the other hand, even if $h$ fails to distinguish the buckets, the adversary does at least as well as random guessing. So overall, it still has a noticeable advantage at the ORE security game.

\trimrm{

We now give the formal proof of Theorem \ref{thm:reid-main}.

\begin{proof}

We construct an example reidentification scheme for $\encthresh$ as follows. The algorithm $\GenTT$ fixes the threshold $t = N/2$ and samples $(\sk^r, \params^r) \getsr \Gen(1^\lambda, 1^\ell)$, yielding a concept $f_{t, r}$. Let $\cD$ be the distribution of $(\params^r, \Enc(\sk^r, m))$ for uniformly random $m \in [N]$. Let $m_1',\dots,m_n'\getsr [N]$, and let $m_1 \le \dots \le m_n$ be the result of sorting the $m_i'$.  Let $m_0 = 0$ and $m_{n+1} = N$. Since $n = \poly(k) \ll N$, these random messages will be well-spaced. In particular, with overwhelming probability, $|m_{i+1} - m_i| > 1$ for every $i$, so we assume this is the case in what follows. $\GenTT$ then sets the samples to be $(x_1 = (\params^r, \Enc(\sk^r, m_1')), \dots, x_n = (\params^r, \Enc(\sk^r, m_n')))$. Let $x_0 = (\params^r, \Enc(\sk^r, m_0))$ be a ``junk'' example.

The algorithm $\TraceTT$ creates buckets $B_i = [m_i, m_{i+1})$. For each $i$, let
\[p_i = \Pr_{m \in B_i, \text{coins of } \Enc}[h(\params^r, \Enc(\sk, m)) = 1].\]
By sampling random choices of $m$ in each bucket, $\TraceTT$ can efficiently compute a good estimate $\hat{p}_i \approx p_i$ for each $i$ (Lemma \ref{lem:prob-est}). It then accuses the least $i$ for which $\hat{p}_{i-1} - \hat{p}_{i} \ge \frac{\gamma}{n}$, and $\bot$ if none is found.

\begin{lemma} \label{lem:prob-est}
Let $K = \frac{8n^2}{\gamma^2}\log(9n/\xi)$. For each $i = 0, \dots, n$, let
\[\hat{p}_i = \frac{1}{K} \sum_{j = 1}^K h(x_j)\]
where $x_j = (\params^r, \Enc(\sk^r, m_j))$ for i.i.d. $m_1, \dots, m_K \getsr B_i$. Then $|\hat{p}_i - p_i| \le \frac{\gamma}{4n}$ for every $i$ with probability at least $1-\xi/4$.
\end{lemma}

\begin{proof}
By a Chernoff bound, the probability that any given $\hat{p}_i$ deviates from $p_i$ by more than $\frac{\gamma}{4n}$ is at most $2\exp(-K\gamma^2/8n^2) \le \frac{\xi}{4(n+1)}$. The lemma follows by a union bound.
\end{proof}

We first verify completeness for this scheme. Let $L$ be a learner for $\encthresh$ using $n$ examples. If the hypothesis $h$ produced by $L$ is $(\frac{1}{2} - \gamma)$-good, then there exists $i_0 < i_1$ such that $p_{i_0} - p_{i_1} \ge 2\gamma$. If this is the case, then there must be an $i$ for which $p_{i-1} - p_{i} \ge \frac{2\gamma}{n}$. Then with probability all but $\xi(n)/2$ over the estimates $\hat{p}_i$, we have $\hat{p}_{i-1} - \hat{p}_{i} \ge \frac{\gamma}{n}$, so some index is accused.

Now we verify soundness. Fix a PPT $L$, and let $j^* \in [n]$. Suppose $L$ violates the soundness of the scheme with respect to $j^*$, i.e.
\[\Pr_{h\getsr L(S_{-j^*}), \text{coins of } \GenTT}[\TraceTT(h) = j^*] > \xi.\]
We will use $L$ to construct an adversary $\cA$ for the ORE scheme that succeeds with noticeable advantage. It suffices to build an adversary for the static (many-challenge) security of ORE, with Theorem~\ref{thm:singlemessage} showing how to convert it to a single-challenge adversary.  This many-challenge adversary is presented as Algorithm \ref{fig:ore-adv}. (While not explicitly stated, the adversary should halt and output a random guess whenever the messages it samples are not well-spaced.)

\begin{algorithm}[h]
\caption{ORE adversary $\cA$}
\label{fig:ore-adv}
\begin{enumerate}
\item Sample $m_1',\dots,m_n'\getsr [N]$, and let $m_1 \le \dots \le m_n$ be the result of sorting the $m_j'$. Let $\pi$ be the permutation on $\{1,\dots,n\}$ such that $m_{\pi(j)}=m'_{j}$.  Let $m_0 = 0$.  Let $i^* = \pi(j^*)$ so that $m_{i^*}=m'_{j^*}$.
\item Construct pairs $(m_L^0, m_L^1)$ and $(m_R^0, m_R^1)$ as follows. Let $B_0 = (m_{i^*-1}, m_{i^*})$ and $B_1 = (m_{i^*}, m_{i^*+1})$. Sample $m_L^0 \le m_L^1$ at random from the same $B_j$, for a random choice of $j \in \{0, 1\}$. Sample $m_R^0 \getsr B_0$ and $m_R^1 \getsr B_1$.
\item Challenge on the pair of sequences $m_0, m_1, \dots, m_{i^*-1}, m_L^1, m_L^2, m_{i^*}, \dots, m_n$ and $m_0, m_1, \dots, m_{i^*-1}, m_R^1, m_R^2, m_{i^*}, \dots, m_n$, receiving ciphertexts $c_1, \dots, c_{i^*}^0, c_{i^*}^1, \dots, c_n.$  For $j\neq j^*$, let $c'_{j}=c_{\pi(j)}$ so that $c'_j$ is an encryption of $m'_j$.
\item Set $t = N/2$ and let 
\begin{align*}
S_{-j^*} &= \big\{(\params^r, c'_1, \chi(m'_1 \le t)), \dots, (\params^r, c'_{j^*-1}, \chi(m'_{j^*-1} \le t)),\\
		 &\hspace{20pt}(\params^r, c_0, 1), (\params^r, c'_{j^*+1}, \chi(m'_{j^*+1} \le t)), \dots, (\params^r, c'_{n}, \chi(m'_{n} \le t))\big\}\\
		 &=\big\{(\params^r, c_{\pi(1)}, \chi(m_{\pi(1)} \le t)), \dots, (\params^r, c_{\pi(j^*-1)}, \chi(m_{\pi(j^*-1)} \le t)),\\
		 &\hspace{20pt}(\params^r, c_0, 1), (\params^r, c_{\pi(j^*+1)}, \chi(m_{\pi(j^*+1)} \le t)), \dots, (\params^r, c_{\pi(n)}, \chi(m_{\pi(n)} \le t))\big\}
\end{align*}
Obtain $h \getsr L(S_{-j^*})$.
\item Guess $b' = 0$ if $h(\params^r, c_{i^*}^0) = h(\params^r, c_{i^*}^1)$. Otherwise guess $b' = 1$.
\end{enumerate}
\end{algorithm}

Let $i^*$ be such that $m_{i^*}=m'_{j^*}$.  With probability at least $\xi$ over the parameters $(\sk^r, \params^r)$, the choice of messages, the choice of the hypothesis $h$, and the coins of $\TraceTT$, there is a gap $\hat{p}_{i^*-1} - \hat{p}_{i^*} \ge \frac{\gamma}{n}$. Hence, by Lemma \ref{lem:prob-est}, there is a gap $p_{i^* - 1} - p_{i^*} \ge \frac{\gamma}{2n}$ with probability at least $\frac{\xi}{2}$.

We now calculate the advantage of the adversary $\cA$. Fix a hypothesis $h$. For notational simplicity, let $p = p_{i^* - 1}$ and let $q = p_{i^*}$. Let $y_0 = h(\params^r, c_{i^*}^0)$ and $y_1 = h(\params^r, c_{i^*}^1)$. Then the adversary's success probability is:

\begin{align*}
\Pr[b' = b] &= \frac{1}{2}(\Pr[y_0 = y_1 | b = 0 ]+ \Pr[y_0 \ne y_1| b = 1]) \\
         &= \frac{1}{2}(\frac{1}{2}(p^2 + (1-p)^2 + q^2 + (1-q)^2) + (1 - pq - (1-p)(1-q))) \\
         &= \frac{1}{2}+ \frac{1}{2} (p-q)^2. \\
\end{align*}

Thus if $p - q \ge \frac{\gamma}{2n}$, then the adversary's advantage is at least $\frac{\gamma^2}{4n^2}$. On the other hand, even for arbitrary values of $p, q$, the advantage is still nonnegative. Therefore, the advantage of the strategy is at least $\frac{\xi \gamma^2}{8n^2} - \negl(k)$ (the $\negl(k)$ term coming from the assumption that the $m_i'$ sampled where distinct), which is a noticeable function of the parameter $k$. This contradicts the static security of the ORE scheme.

\end{proof}

\subsection{The SQ Learnability of $\encthresh$} \label{sec:sq}

The statistical query (SQ) model is a natural restriction of the PAC model by which a learner is able to measure statistical properties of its examples, but cannot see the individual examples themselves. We recall the definition of an SQ learner.

\begin{definition}[SQ learning \cite{Kearns98}]
Let $c : X \to \{0, 1\}$ be a target concept and let $\cD$ be a distribution over $X$. In the SQ model, a learner is given access to a \emph{statistical query oracle} $\stat(c, \cD)$. It may make queries to this oracle of the form $(\psi, \tau)$, where $\psi: X \times \{0, 1\} \to \{0, 1\}$ is a query function and $\tau \in (0, 1)$ is an error tolerance. The oracle $\stat(c, \cD)$ responds with a value $v$ such that $|v - \Pr_{x \in \cD}[\psi(x, c(x)) = 1]| \le \tau$. The goal of a learner is to produce, with probability at least $1-\beta$, a hypothesis $h : X \to \{0, 1\}$ such that $\error_{\cD}(c, h) \le \alpha$. The query functions must be efficiently evaluable, and the tolerance $\tau$ must be lower bounded by an inverse polynomial in $k$ and $1/\alpha$.

The \emph{query complexity} of a learner is the worst-case number of queries it issues to the statistical query oracle. An SQ learner is efficient if it also runs in time polynomial in $k, 1/\alpha, 1/\beta$.
\end{definition}

Feldman and Kanade \cite{FeldmanKa12} investigated the relationship between query complexity and computational complexity for SQ learners. They exhibited a concept class $\cC$ which is efficiently PAC learnable and SQ learnable with polynomially many queries, but assuming $\mathbf{NP} \neq \mathbf{RP}$, is not efficiently SQ learnable. Concepts in this concept class take the form
\[g_{\phi, y}(x, x') = \begin{cases}
\operatorname{PAR}_y(x') & \text{ if } x = \phi \\
0 & \text{ otherwise.}
\end{cases}\]
Here, $\operatorname{PAR}_y(x')$ is the inner product of $y$ and $x'$ modulo $2$. The concept class $\cC$ consists of $g_{\phi, y}$ where $\phi$ is a satisfiable 3-CNF formula and $y$ is the lexicographically first satisfying assignment to $\phi$. The efficient PAC learner for parities based on Gaussian elimination shows that $\cC$ is also efficiently PAC learnable. It is also (inefficiently) SQ learnable with polynomially many queries: either the all-zeroes hypothesis is good, or an SQ learner can recover the formula $\phi$ bit-by-bit and determine the satisfying assignment $y$ by brute force. On the other hand, because parities are information-theoretically hard to SQ learn, the satisfying assignment $y$ remains hidden to an SQ learner unless it is able to solve 3-SAT.

In this section, we show that the concept class $\encthresh$ shares these properties with $\cC$. Namely, we know that $\encthresh$ is efficiently PAC learnable and because it is not efficiently privately learnable, it is not efficiently SQ learnable \cite{BlumDwMcNi05}. We can also show that $\encthresh$ has an SQ learner with polynomial query complexity. Making this observation about $\encthresh$ is of interest because the hardness of SQ learning $\encthresh$ does not seem to be related to the (information-theoretic) hardness of SQ learning parities.

\begin{proposition}
The concept class $\encthresh$ is (inefficiently) SQ learnable with polynomially many queries.
\end{proposition}

As with $\cC$ there are two cases. In the first case, the target distribution places nearly zero weight on examples with $\params = \params^r$, and so the all-zeroes hypothesis is good. In the second case, the target distribution places noticeable weight on these examples, and our learner can use statistical queries to recover the comparison parameters $\params^r$ bit-by-bit. Once the public parameters are recovered, our learner can determine a corresponding secret key by brute force. Lemma \ref{lem:sk-recovery} below shows that any corresponding secret key -- even one that is not actually $\sk^r$ -- suffices. The learner can then use binary search to determine the threshold value $t$.

\begin{proof}
Let $f_{t, r}$ be the target concept, $\cD$ be the target distribution, and $\alpha$ be the target error rate. With the statistical query $(x \times b \mapsto b, \alpha/4)$, we can determine whether the all-zeroes hypothesis is accurate. That is, if we receive a value that is less than $\alpha/2$, then $\Pr_{x \in \cD}[f_{t, r}(x) = 1] \le \alpha$. If not, then we know that $\Pr_{x \in \cD}[f_{t, r}(x) = 1] \ge \alpha/4$, so $\cD$ places significant weight on examples prefixed with $\params^r$. Suppose now that we are in the latter case.

Let $m = |\params|$. For $i = 1, \dots, m$, define $\psi_i(\params, c, b) = 1$ if $\params_i = 1$ and $b = 1$, and $\psi_i(\params, c, b) = 0$ otherwise. Then by asking the queries $(\psi_i, \alpha/16)$, we can determine each bit $\params^r_i$ of $\params^r$.

Now by brute force search, we determine a secret key $\sk$ for which $(\sk, \params^r) \in \operatorname{Range}(\Gen)$. The recovered secret key $\sk$ may not necessarily be the same as $\sk^r$. However, the following lemma shows that $\sk$ and $\sk^r$ are functionally equivalent:

\begin{lemma} \label{lem:sk-recovery}
Suppose $(\Gen, \Enc, \Dec, \Comp)$ is a strongly correct ORE scheme. Then for any pair $(\sk_1, \params), (\sk_2, \params) \in \operatorname{Range}(\Gen)$, we have that $\Dec_{\sk_1}(c) = \Dec_{\sk_2}(c)$ for all ciphertexts $c$.
\end{lemma}

With the secret key $\sk$ in hand, we now conduct a binary search for the threshold $t$. Recall that we have an estimate $v$ for the weight that $f_{t, r}$ places on positive examples, i.e. $|v - \Pr_{x \in \cD}[f_{t, r}(x) = 1]| \le \alpha/4$. Starting at $t_1 = N/2$, we issue the query $(\varphi_1, \alpha/4)$ where $\varphi_1(\params, c, b) = 1$ iff $\params = \params^r$ and $\Dec(\sk, c) < t$. Let $h_{t_1}$ denote the hypothesis
\[h_{t_1}(\params, c) = \begin{cases}
1 & \text{if } (\params = \params^r) \land (\Dec(\sk, c) \ne \bot) \land (\Dec(\sk, c) < t_1) \\
0 & \text{otherwise.}
\end{cases}\]
Thus, the query $(\varphi_1, \alpha/4)$ approximates the weight $h_{t_1}$ places on positive examples.  Let the answer to this query be $v_1$. If $|v_1 - v| \le \alpha / 2$, then we can halt and output the good hypothesis $h_{t_1}$. Otherwise, if $v_1 < v - \alpha/2$, we set the next threshold to $t_2 = 3N/4$, and if $v_1 > v + \alpha/2$, we set the next threshold to $t_2 = N/4$. We recurse up to $\log N = \ell = \poly(k)$ times, yielding a good hypothesis for $f_{t, r}$.
\end{proof}

\begin{proof}[Proof of Lemma \ref{lem:sk-recovery}]
Suppose the lemma is not true. First suppose that there exists a ciphertext $c$ such that $\Dec(\sk_1, c) = p_1 < p_2 = \Dec(\sk_2, c)$. Let $c' \in \Enc(\sk_1, p_2)$. Then by strong correctness applied to the parameters $(\sk_1, \params)$, we must have $\Comp(\params, c, c') =$ ``$<$''. Now by strong correctness applied to $(\sk_2, \params)$, we must have $\Dec(\sk_2, c') > p_2$. Thus, $p_1 < \Dec(\sk_1, c') = p_2 < \Dec(\sk_2, c')$. Repeating this argument, we obtain a contradiction because the message space is finite.

Now suppose instead that there is a ciphertext $c$ for which $\Dec(\sk_1, c) = p \in [N]$, but $\Dec(\sk_2, c) = \bot$. Let $c' \in \Enc(\sk_1, p')$ for some $p' > p$. Then $\Comp(\params, c, c') =$ ``$<$'' by strong correctness applied to $(\params, \sk_1)$. But $\Comp(\params, c, c') =$ ``$\bot$'' by strong correctness applied to $(\params, \sk_2)$, again yielding a contradiction. 
\end{proof}

}

\trimrm{

\section{ORE with Strong Correctness}

\label{sec:construction}

We now explain how to obtain ORE with strongly correct comparison, as all prior ORE schemes only satisfy the weaker notion of correctness.  The lack of strong correctness is easiest to see with the scheme of Boneh et al.~\cite{BLRSZZ15}.  The protocol is built from current multilinear map constructions, which are noisy.  If the noise terms grow too large, the correctness of the multilinear map is not guaranteed.  The comparison function in~\cite{BLRSZZ15} is computed by performing multilinear operations, and for correctly generated ciphertexts, the operations will give the right answer.  However, there exist ciphertexts, namely those with very large noise, for which the comparison function gives an incorrect output.  The result is that the comparison operation is not guaranteed to be consistent with decrypting the ciphertexts and comparing the plaintexts. 

As described in the introduction, we give a generic conversion from any ORE scheme with weakly correct comparison into a strongly correct scheme.  We simply modify the encryption algorithm by adding a non-interactive zero-knowledge (NIZK) proof that the resulting ciphertext is well-formed.  Then the decryption and comparison procedures check the proof(s), and only output a non-$\bot$ result (either decryption or comparison) if the proof(s) are valid. 

\paragraph{Instantiating our scheme.}  In our construction, we need the (weak) correctness of the underlying ORE scheme to hold with probability one.  However, the existing protocols only have correctness with overwhelming probability, so some minor adjustments need to be made to the protocols.  This is easiest to see in the ORE scheme of Boneh et al.~\cite{BLRSZZ15}.  The Boneh et al. scheme uses noisy multilinear maps~\cite{GGH13} which may introduce errors.  Therefore, the protocol described in~\cite{BLRSZZ15} only achieves the (weak) correctness property with overwhelming probability, whereas we will require (weak) correctness with probability 1 for the conversion.  However, it is straightforward to generate the parameters for the protocol in such a way as to completely eliminate errors.  Essentially, the parameters in the protocol have an error term that is generated by a (discrete) Gaussian distribution, which has unbounded support.   Instead, we truncate the Gaussian, resulting in a noise distribution with bounded support.  By truncating sufficiently far from the center, the resulting distribution is also statistically close to the full Gaussian, so security of the protocol with truncated noise follows from the security of the protocol with un-truncated noise.  By truncating the noise distribution, it is straightforward to set parameters so that no errors can occur.  

It is similarly straightforward to modify current obfuscation candidates, which are also built from multilinear maps, to obtain perfect (weak) correctness by truncating the noise distributions.  Thus, our scheme has instantiations using multilinear maps or iO.

\subsection{Conversion from Weakly Correct ORE} 

We describe our generic conversion from an order-revaling encryption scheme with weak correctness using NIZKs.  We will need the following additional tools:

\paragraph{Perfectly binding commitments.} A perfectly binding commitment $\Com$ is a randomized algorithm with two properties.  The first is perfect binding, which states that if $\Com(m;r)=\Com(m';r')$, then $m=m'$.  The second requirement is computational hiding, which states that the distributions $\Com(m)$ and $\Com(m')$ are computationally indistinguishable for any messages $m,m'$.  Such commitments can be built, say, from any injective one-way function.

\paragraph{Perfectly sound NIZK.}  A NIZK protocol consists of three algorithms:

\begin{itemize}
	\item $\Setup(1^\lambda)$ is a randomized algorithm that outputs a common reference string $\crs$.
	\item $\Prove(\crs,x,w)$ takes as input a common reference string $\crs$, an NP statement $x$, and a witness $w$, and produces a proof $\pi$.
	\item $\Ver(\crs,x,\pi)$ takes as input a common reference string $\crs$, statement $x$, and a proof $\pi$, and outputs either $\accept$ or $\reject$.
\end{itemize}

We make three requirements for a NIZK:
\begin{itemize}
\item {\bf Perfect Completeness.}  For all security parameters $\lambda$ and any true statement $x$ with witness $w$, \[\Pr[\Ver(\crs,x,\pi)=\accept:\crs\gets\Setup(1^\lambda);\pi\gets\Prove(\crs,x,w)]=1.\]
\item {\bf Perfect Soundness.}  For all security parameters $\lambda$, any \emph{false} statement $x$ and any (invalid) proof $\pi$, \[\Pr[\Ver(\crs,x,\pi)=\accept:\crs\gets\Setup(1^\lambda)]=0.\]
\item {\bf Computational Zero Knowledge.} There exists a simulator $\cS_1,\cS_2$ such that for any computationally bounded adversary $\cA$, the quantity

\[\|\Pr[\cA^{\Prove(\crs,\cdot,\cdot)}(\crs)=1:\crs\gets\Setup(1^\lambda)]-\Pr[\cA^{Sim(\crs,\tau,\cdot,\cdot)}(\crs)=1:(\crs,\tau)\gets \cS_1(1^\lambda)]\|\]

is negligible, where $Sim(\crs,\tau,x,w)$ outputs $\cS_2(\crs,\tau,x)$ if $w$ is a valid witness for $x$, and $Sim(\crs,\tau,x,w)=\bot$ if $w$ is invalid.
\end{itemize}

NIZKs satisfying these requirements can be built from bilinear maps~\cite{GOS12}.

\medskip

\subsubsection{The Construction}

We now give our conversion.  Let $(\Setup,\Prove,\Ver)$ be a perfectly sound NIZK and $(\Gen',\Enc',\Dec',\Comp')$ and ORE with \emph{weakly} correct comparison.  We will assume that $\Enc'$ is deterministic; if not, we can derandomize $\Enc'$ using a pseudorandom function.  Let $\Com$ be a perfectly binding commitment.   We construct a new ORE scheme $(\Gen,\Enc,\Dec,\Comp)$ with \emph{strongly} correct comparison:

\begin{itemize}
	\item $\Gen(1^\lambda,1^\ell)$: run $(\sk',\params')\gets\Gen'(1^\lambda,1^\ell)$.  Let $\sigma=\Com(\sk;r)$ for randomness $r$, and run $\crs\gets\Setup(1^\lambda)$.  Then the secret key is $\sk=(\sk',r,\crs)$ and the public parameters are $\params=(\params',\sigma,\crs)$.
	\item $\Enc(\sk,m)$: Compute $c'=\Enc'(\sk',m)$.  Let $x_{c'}$ be the statement $\exists \hat{m},\hat{\sk}',\hat{r}:\sigma=\Com(\hat{\sk}',\hat{r})\wedge c'=\Enc'(\hat{\sk}',\hat{m})$.  Run $\pi_{c'}=\Prove(\crs,x_{c'},\;(m,\sk',r)\;)$.  Output the ciphertext $c=(c',\pi_{c'})$.
	\item $\Dec(\sk,c)$: Write $c=(c',\pi_{c'})$.  If $\Ver(\crs, x_{c'},\pi_{c'})=\reject$, output $\bot$.  Otherwise, output $m=\Dec'(\sk',c')$.
	\item $\Comp(\params,c_0,c_1)$; white $c_b=(c_b',\pi_{c_b'})$ and $\params=(\params',\sigma,\crs)$. If $\Ver(\crs,x_{c_b'},\pi_{c_b'})=\reject$ for either $b=0,1$, then output $\bot$.  Otherwise, output $\Comp'(\params',c_0',c_1')$. 
\end{itemize}

\paragraph{Correctness.}  Notice that, for each plaintext $m$, the ciphertext component $c'=\Enc'(\sk',m)$ is the \emph{unique} value such that $\Dec(\sk,(c',\pi))=m$ for some proof $\pi$.  Moreover, the completeness of the zero knowledge proof implies that $\Enc(\sk,m)$ outputs a valid proof.  Decryption correctness follows.

For strong comparison correctness, consider two ciphertexts $c_0,c_1$ where $c_b=(c_b',\pi_{c_b'})$.  Suppose both proofs $\pi_{c_b'}$ are valid, which means that verification passes when running $\Comp$ and so $\Comp(\params,c_0,c_1)=\Comp'(\params',c_0',c_1')$.  Verification also passes when decrypting $c_b$, and so $\Dec(\sk,c_b)=\Dec'(\sk',c_b')$.

Since the proofs are valid, $c_b'=\Enc'(\sk',m_b)$ for some $m_b$ for both $b=0,1$.  The weak correctness of comparison for $(\Gen',\Enc',\Dec',\Comp')$ implies that $\Comp'(\params',c_0',c_1')=\Comp_{plain}(m_0,m_1)$.  The decryption correctness of $(\Gen',\Enc',\Dec',\Comp')$ then implies that $\Dec(\sk',c_b')=m_b$, and therefore $\Dec(\sk,c_b)=m_b$.  Thus $\Comp_{ciph}(\sk,c_0,c_1)=\Comp_{plain}(m_0,m_1)$.  Putting it all together, $\Comp(\params,c_0,c_1)=\Comp_{ciph}(\sk,c_0,c_1)$, as desired. 

Now suppose one of the proofs $\pi_{c_b'}$ are invalid.  Then $\Comp(\params,c_0,c_1)=\bot$ and $\Dec(\sk,c_b)=\bot$.  This means $\Comp_{ciph}(\sk,c_0,c_1)=\bot=\Comp(\params,c_0,c_1)$, as desired.

\paragraph{Security.}  To prove security, we first use the zero-knowledge simulator to simulate the proofs $\pi_c'$ without using a witness (namely, the secret decryption key).  Then we use the hiding property of the commitment to replace $\sigma$ with a commitment to 0.  At this point, the entire game can be simulated using an $\Enc'$ oracle, and so the security reduces to the security of $\Enc'$.  

\begin{theorem}
	If $(\Gen',\Enc',\Dec',\Comp')$ is a (statically) secure ORE, $(\Setup,\Prove,\Ver)$ is computationally zero knowledge, and $\Com$ is computationally hiding, then $(\Gen,\Enc,\Dec,\Comp)$ is a statically secure ORE.
\end{theorem}

\begin{proof} We will prove security through a sequence of hybrids.  Let $\cA$ be an adversary with advantage $\epsilon$ in breaking the static security of $(\Gen,\Enc,\Dec,\Comp)$.

	\paragraph{Hybrid 0.} This is the real experiment, where $\sigma\gets\Com(\sk)$, $\crs\gets\Setup(1^\lambda)$, and the proofs $\pi_{c'}$ are answered using $\Prove$ and valid witnesses.  $\cA$ has advantage $\epsilon$ in distinguishing the left and right ciphertexts.
	
	\paragraph{Hybrid 1.} This is the same as {\bf Hybrid 0}, except that $\crs$ is generated as $(\crs,\tau)\gets\cS_1(1^\lambda)$, and all proofs are generated using $\cS_2(\crs,\tau,\cdot)$.  The zero knowledge property of $(\Setup,\Prove,\Ver)$ shows that this is indistinguishable from {\bf Hybrid 0}.
	
	\paragraph{Hybrid 2.} This is the same as {\bf Hybrid 1}, except that $\sigma\gets\Com(0)$.  Since the randomness for computing $\sigma$ is not needed for simulation, this change is undetectable using the hiding of $\Com$.  
	
	\medskip
	
	Thus the advantage of $\cA$ in {\bf Hybrid 2} is at least $\epsilon-\negl$ for some negligible function $\negl$.  Now consider the following adversary $cB$ that attempts to break the security of $(\Gen',\Enc',\Dec',\Comp')$.  $\cB$ simulates $\cA$, and forwards the message sequences $m_1^{(L)}<m_2^{(L)}<\dots<m_q^{(L)}$ and $m_1^{(R)}<m_2^{(R)}<\dots<m_q^{(R)}$ produced by $\cA$ to its own challenger.  In response, it receives $\params'$, and ciphertexts $c_i'$, where $c_i'$ encrypts either $m_i^{(L)}$ if $b=0$ or $m_i^{(R)}$ if $b=1$, for a random bit $b$ chosen by the challenger.  
	
	$\cB$ now generates $\sigma\gets\Com(0)$ and $(\crs,\tau)\gets\cS_1(1^\lambda)$, and lets $\params=(\params',\sigma,\crs)$.  It also computes $\pi_{c_i'}\gets\cS_2(\crs,\tau,x_{c_i'})$, and defines $c_i=(c_i',\pi_{c_i'})$, and gives $\params$ and the $c_i$ to $\cA$.  Finally when $\cA$ outputs a guess $b'$ for $b$, $\cB$ outputs the same guess $b'$.
	
	We see that the view of $\cA$ as a subroutine of $\cB$ is exactly the same view as in {\bf Hybrid 2}.  Thus, $b'=b$ with probability at least $\epsilon-\negl$.  The security of $(\Gen',\Enc',\Dec',\Comp')$ implies that this quantity, and hence $\epsilon$, must be negligible.  Thus $\cA$ must have negligible advantage in breaking the security of $(\Gen,\Enc,\Dec,\Comp)$.
	
\end{proof}

\section{A Separation for Representation Learning}

\label{sec:proper}
In this section, we show how to construct a concept class $\sig$ that separates efficient \emph{representation} learning from efficient private representation learning, assuming only the existence of one-way functions. Here by ``representation learning'' we mean a restricted form of proper learning where a learner must output a particular representation (i.e. encoding) of a hypothesis $h$ in the concept class $\cC$. As with proper learning, this is a natural syntactic restriction to place on a learner: for instance, if one wants to learn linear threshold functions (LTF), it makes sense to require a learner to produce the actual coefficients of an LTF, rather than an arbitrary circuit that happens to compute an LTF.

The construction is based on the following elegant idea due to Kobbi Nissim \cite{Nissim14}. Suppose $H: D \to R$ is a cryptographic hash function with the property that given $x_1, \dots, x_n$ with $y = H(x_1) = \dots = H(x_n)$, it is infeasible for an efficient adversary to find another $x$ for which $H(x) = y$. Consider the concept class $\hashpoint$ consisting of the concepts
\[f_x(x') = \begin{cases}
1 & \text{ if } H(x) = H(x') \\
0 & \text{ otherwise.}
\end{cases}\]
for every $x \in R$. The representation of a concept $f_x$ is the point $x$. The concept class $\hashpoint$ is very easy to learn (by representation) without privacy: a learner can identify any positive example $x_i$ and output the representation $x_i$. Since $H(x_i) = H(x)$, the concept $f_{x_i}$ is actually equal to the target concept $f_x$. On the other hand, a learner that identifies an index $x^*$ for which $f_{x^*} = f_x$ cannot be differentially private, since the security of the hash function means it is infeasible to produce such an $x^*$ that is not present in the sample.

Note that this argument breaks down if one tries to show that $\hashpoint$ is not privately properly learnable. While it is infeasible to privately produce a representation $x^*$ for which $f_{x^*}$ is a good hypothesis, the hypothesis $h(x) = \chi(H(x) = h(x_i))$ is equal as a function to every good $f_{x^*}$. Moreover, this hypothesis can be constructed privately as long as the sample contains sufficiently many positive examples.


We make this discussion formal by constructing a concept class $\sig$ based on \emph{super-secure digital signature schemes}, which can be constructed from one-way functions. Our use of signatures to derive hardness results for private proper learning is very analogous to prior hardness results for synthetic data generation \cite{DNRRV09, UllmanVa11}.

\begin{definition}
A \emph{digital signature scheme} is a triple of algorithms $(\Gen, \Sign, \Ver)$ where
\begin{itemize}
\item $\Gen(1^\lambda)$ produces a key pair $(\sk, \vk)$.
\item $\Sign(\sk, m)$ takes the private signing key $\sk$ and a message $m \in \{0, 1\}^*$ and produces a signature $\sigma$ for the message $m$.
\item $\Ver(\vk, m, \sigma)$ takes the public verification key $\vk$, a message $m$, and a signature $\sigma$, and (deterministically) outputs a bit indicating whether $\sigma$ is a valid signature for $m$.
\end{itemize}
The correctness property of a digital signature scheme is that for every $(\sk, \vk) \getsr \Gen(1^\lambda)$, every message $m \in \{0,1\}^*$, and every signature $\sigma \getsr \Sign(\sk, m)$, we have $\Ver(\vk, m, \sigma) = 1$.
\end{definition}

\begin{definition}
A digital signature scheme is \emph{super-secure under adaptive chosen-plaintext attacks} if all efficient adversaries $\cA$ win the following weak forgery game with negligible probability:
\begin{itemize}
\item The challenger samples $(\sk, \vk) \getsr \Gen(1^\lambda)$.
\item The adversary $\cA$ is given $\vk$ and oracle access to $\Sign(\sk, \cdot)$. It adaptively queries the signing oracle, obtaining a sequence of message-signature pairs $A$. It then outputs a forgery $(m^*, \sigma^*)$.
\item The value of the game is 1 iff $\Ver(\vk, m^*, \sigma^*) = 1$ and $(m^*, \sigma^*) \notin A$.
\end{itemize}
\end{definition}

It is known that super-secure digital signature schemes can be constructed from one-way functions \cite{NaorYu89, Rompel90, KatzKo05, Goldreich04}.


We now describe our concept class $\sig$. Let $(\Gen, \Sign, \Ver)$ be a super-secure digital signature scheme. We define a concept class $\sig$ as follows. Fix the message length $\ell$. For every $(\vk, m, \sigma)$ with $m \in \{0, 1\}^\ell$ and $\Ver(\vk, m, \sigma) = 1$, define the concept
\[f_{\vk, m, \sigma}(\vk', m', \sigma') = \begin{cases}
1 &\text{ if } (\vk = \vk') \land (\Ver(\vk, m', \sigma') = 1) \\
0 & \text{ otherwise.}
\end{cases}\]
For convenience, we also include the all-zeroes hypothesis in $\sig$, with representation $\bot$.

\begin{theorem}
Let $\alpha, \beta > 0$. There exists a proper PAC learning algorithm $L$ for the concept class $\sig$ achieving error $\alpha$ and confidence $1-\beta$. Moreover, $L$ is efficient (running in time polynomial in the parameters $k, 1/\alpha, \log(1/\beta)$).
\end{theorem}

\begin{algorithm}[h]
\caption{Learner $L$ for $\sig$}
\label{fig:pac-learner}
\begin{enumerate}
\item Request examples $\{((\vk_1', m_1', \sigma_1'), b_1), \dots, ((\vk_n', m_n', \sigma_n'), b_n)\}$ for $n = \lceil\log(1/\beta)/\alpha\rceil$.
\item Identify an $i$ for which $b_i = 1$ and return the representation $(\vk_i', m_i', \sigma_i')$. If no such $i$ exists, return $\bot$ representing the all-zeroes hypothesis.
\end{enumerate}
\end{algorithm}

\begin{proof}

Fix a target concept $f_{\vk, m, \sigma} \in \sig_k$ and a distribution $\cD$ on examples. Let $\pos$ denote the set of examples $(\vk', m', \sigma')$ on which $f_{\vk, m, \sigma}(\vk', m', \sigma') = 1$. We divide the analysis of the learner into three cases based on the weight $\cD$ places on the sets $\pos$.

\paragraph{Case 1:} $\cD$ places at least $\alpha$ weight on $\pos$. Then $L$ receives a positive example with probability at least $1 - (1-\alpha)^n \ge 1 - \beta$, and is thus able to identify a concept that equals the target concept.

\paragraph{Case 2:} $\cD$ places less than $\alpha$ weight on $\pos$. If $L$ gets a positive example, then the analysis of Case 1 applies. Otherwise, the all-zeroes hypothesis is $\alpha$-good.

\end{proof}

We now prove the hardness of properly privately learning $\sig$ by constructing an example reidentification scheme for properly learning this concept class. Our example reidentification scheme yields a hard distribution even when the error parameter $\alpha$ is taken to be inverse-polynomially close to $1$.

\begin{theorem} \label{thm:reid-proper}
Let $\gamma(n)$ and $\xi(n)$ be noticeable functions. Let $(\Gen, \Sign, \Ver)$ be a super-secure digital signature scheme. Then there exists an (efficient) $(\alpha = 1 - \gamma, \xi)$-example reidentification scheme $(\GenTT, \TraceTT)$ for representation learning the concept class $\sig$.
\end{theorem}

We now give the proof of Theorem \ref{thm:reid-proper}.

\begin{proof}

We construct an example reidentification scheme for $\sig$ as follows. The algorithm $\GenTT$ samples $(\sk, \vk) \getsr \Gen(1^\lambda)$, a message $m \in \{0, 1\}^\ell$, and a signature $\sigma \getsr \Sign(\sk, m)$, yielding a concept $f_{\vk, m, \sigma}$. Let $\cD$ be the distribution of $(\vk, m, \Sign(\sk, m))$ for random $m \getsr \{0, 1\}^\ell$. $\GenTT$ then samples $x_0, x_1, \dots, x_n$ i.i.d. from $\cD$. Given a representation $(\vk^*, m^*, \sigma^*)$, the algorithm $\TraceTT$ simply identifies an index $i$ for which $x_i = (\vk^*, m^*, \sigma^*)$, and outputs $\bot$ if none is found.

We first verify completeness for this scheme. Let $L$ be a learner for $\sig$ using $n$ examples. If the representation $(\vk^*, m^*, \sigma^*)$ produced by $L$ represents an $(1 - \gamma)$-good hypothesis, then it must be the case that $\vk^* = \vk$ and $\Ver(\vk, m^*, \sigma^*) = 1$. Thus, if $L$ violates the completeness condition, it can be used to construct the weak forgery adversary $\cA$ (Figure \ref{fig:sig-adv}) that succeeds with noticeable probability $\xi$.

\begin{algorithm}[H]
\caption{Weak forgery adversary $\cA$}
\label{fig:sig-adv}
\begin{enumerate}
\item Query the signing oracle on random messages $m'_1, \dots, m'_n \getsr \{0, 1\}^\ell$, obtaining signatures $\sigma'_1, \dots, \sigma'_n$.
\item Run $L$ on the labeled examples $((\vk, m'_1, \sigma'_1), 1), \dots, ((\vk, m'_n, \sigma'_n), 1)$, obtaining a representation $(m^*, \sigma^*)$.
\item Output the forgery $(m^*, \sigma^*)$.
\end{enumerate}
\end{algorithm}

Now we verify soundness for the scheme. Observe that for any $i$, the sample $S_{-i}$ contains no information about message $m_i$. Therefore, the learner has a $2^{-\ell} = \negl(k)$ probability at producing a representation containing message $m_i$, proving soundness.

\end{proof}

}

\paragraph{Acknowledgements.} We gratefully acknowledge Kobbi Nissim and Salil Vadhan for helpful discussions about this work, and also thank Salil Vadhan for suggestions on its presentation.

\bibliographystyle{alpha}
\bibliography{references}

\end{document}